\documentclass[11pt]{article}

	\usepackage{enumerate,hyperref,color,enumitem,layouts,calc,siunitx}
	\usepackage{cite}
	\usepackage{bm,amsmath,amssymb,mathtools,physics}
	\usepackage{graphicx}
	\usepackage{algorithm,algpseudocode}\usepackage[font=small,labelfont=bf,figurename=Figure,labelsep=period]{caption}
	\usepackage[margin=1in]{geometry}
	\usepackage{float,pdflscape,authblk,setspace,lineno,cleveref}
    \usepackage{combelow}
    \usepackage{amsthm}
    \newtheorem{theorem}{Theorem}
    \newtheorem{corollary}{Corollary}
    
	\newcommand{\refereecomment}[2]{\label{#2}\marginpar{\fbox{\texttt{#1}}}}

    \newcommand{\killpunct}[1]{}

\title{Structural identifiability analysis of linear reaction-advection-diffusion processes in mathematical biology}
\date{\today}

	\author[1]{Alexander P Browning}
    \author[1,2]{Maria Ta\cb{s}c\u{a}}
    \author[1]{Carles Falc\'o}
	\author[1]{Ruth E Baker}
	
	\affil[1]{Mathematical Institute, University of Oxford, Oxford, United Kingdom}
	\affil[2]{Somerville College, University of Oxford, Oxford, United Kingdom}

\begin{document}

\maketitle

	\vfill
	\renewcommand{\abstractname}{Abstract}
	\begin{abstract}
		\noindent
		Effective application of mathematical models to interpret biological data and make accurate predictions often requires that model parameters are identifiable. Approaches to assess the so-called structural identifiability of models are well-established for ordinary differential equation models, yet there are no commonly adopted approaches that can be applied to assess the structural identifiability of the partial differential equation (PDE) models that are requisite to capture spatial features inherent to many phenomena. The differential algebra approach to structural identifiability has recently been demonstrated to be applicable to several specific PDE models. In this brief article, we present general methodology for performing structural identifiability analysis on partially observed reaction-advection-diffusion (RAD) PDE models that are linear in the unobserved quantities. We show that the differential algebra approach can always, in theory, be applied to such models. Moreover, despite the perceived complexity introduced by the addition of advection and diffusion terms, identifiability of spatial analogues of non-spatial models cannot decrease in structural identifiability. We conclude by discussing future possibilities and the computational cost of performing structural identifiability analysis on more general PDE models.
	\end{abstract}
	\vfill

	\renewcommand{\abstractname}{Keywords}
	\begin{abstract}
		\noindent 
		\centering
		structural identifiability, reaction-diffusion, advection, cell migration, \\chemotaxis, partial differential equations
	\end{abstract}
	\vspace{1cm}
	\vfill

\section{Introduction}

Mathematical models play an irreplaceable role in the interpretation of biological data. Model parameters are now routinely used to objectively quantify observed behaviour and characterise behaviours that cannot be directly measured \cite{Liepe:2014,Gabor:2015dr}. The question of whether it is possible, in theory, for model parameters to be uniquely identified given a specified mathematical model and a specified set of observed quantities, is referred to as \textit{structural identifiability} \cite{Bellman.1970,Walter:1981,Walter:1987,Raue:2009,Chis:2011}. Assessing the structural identifiability of a model and observation process can provide vital insights that guide data collection before experiments have been conducted \cite{Villaverde:2016sg}; establish identifiable parameter combinations that resolve issues of non-identifiability either through reparameterisation or additional experimentation \cite{Meshkat.2009}; and provide confidence in predictions drawn from calibrated mathematical models \cite{Eisenberg.2017}.

Tools for assessing structural identifiability are well established for deterministic, ordinary differential equation (ODE) models \cite{Ljung:1994,Bellu:2007,Raue:2014,Barreiro.2023hg}. The advent of open-source and even web-based software \cite{Hong:2019uq} to automate otherwise the tedious analysis has ingrained questions related to structural identifiability into the inference process for ODE models \cite{Barreiro.2023hg}. Methods originate with differential-algebra-based approaches \cite{Meshkat:2014}, which for linear or polynomial systems assess identifiability through a so-called input-output relation: a set of monic polynomials in the derivatives of observed variables, the coefficients of which form the set of identifiable parameter combinations. Differential algebra approaches are trivial if all variables in a system are observed, although questions of structural identifiability very often relate to partially observed systems, where only measurements of a subset of state variables are available \cite{Bellman.1970,Villaverde:2016sg,Wieland.2021}. For such \textit{partially observed} systems, the algorithmic complexity of differential algebra approaches comes from the reduction of a high-dimensional (in the number of state variables), lower-derivative-order, system to a set of higher-derivative-order polynomials that include only the observed quantities \cite{Ljung:1994}. Several alternative approaches for structural identifiability have been since been established; notably those based on Taylor series, generating series and Lie derivatives \cite{Chis:2011,Ligon.2017qw}, and similarity transforms \cite{Vajda:1989,Chis:2011}. Many of these more modern approaches are more broadly applicable to analytic systems, and significantly more computationally efficient than those based on differential algebra.

However, many forms of biological data are inherently spatial, and therefore not well-described by ODE models \cite{Murray:2002}: data relating to cell migration \cite{Simpson:2020,Liu.2023} or diffusive processes \cite{Kicheva.2007da,Romanova.2022sd}, for example. Yet, tools for assessing the structural identifiability of the partial differential equation (PDE) models that capture spatial heterogeneity remain almost entirely undeveloped. Recent work has demonstrated the application of the differential algebra method to assess structural identifiability of a class of first-order age structured PDE models \cite{Renardy.2022}, and a system comprising a single diffusive species \cite{Ciocanel.2023}. Significantly, Renardy et al. \cite{Renardy.2022} demonstrate that the differential-algebra approach may be applicable more generally to systems of linear first-order PDE models. However, structural identifiability of more generic systems that contain both first and second order spatial derivatives (hereafter referred to as reaction-advection-diffusion (RAD) systems) have not, to the best of our knowledge, been explored. In particular, it has remained unclear whether the differential algebra approach can always be applied to PDE models, or whether such analysis of PDE models is more restrictive. Even more basic questions, including whether the apparent additional complexity of including a diffusion term exacerbates or alleviates issues relating to structural identifiability, remain unanswered. Given the prevalence of spatial data in biology, in no small part due to advancements in microscopy and imaging technologies, it is imperative that tools for structural identifiability are developed for the PDE models that account for inherently spatial behaviour.

In this brief article, we present a methodology and a series of general results for the structural identifiability of linear RAD PDE models subject to partial observation. We also demonstrate that our methodology can be applied to any non-linear RAD PDE model that is linear in the unobserved quantities: we refer to such models as semi-linear. For a more thorough and formal introduction to structural identifiability and the differential algebra approach, we direct the reader to \cite{Margaria.2001,Renardy.2022}. We present our results through example, first presenting a didactic guide to the procedure for a two-state model of cell proliferation and migration inspired by common scratch assay experiments (\cref{fig1}) \cite{Simpson:2020}. Next, in \Cref{generic_two_state}, we present more general methodology and general results for a generic two-state RAD model; these results are then compared to the corresponding ODE model before the role of the initial condition---which differs significantly to the ODE case---is discussed in \Cref{role_of_ic}. In \Cref{generic_linear_systems} we formally present results for generic linear systems comprising an arbitrary number of states. Finally, in \Cref{semi-linear-systems} we present an example set of results for two semi-linear systems: a model of bacterial chemotaxis, and a three-state nonlinear analogue of the scratch assay model. Avenues for future work and key features that distinguish the structural identifiability problem for ODE models from those of PDE models are discussed in \Cref{conclusion}. 

	\begin{figure}
		\centering
		\includegraphics{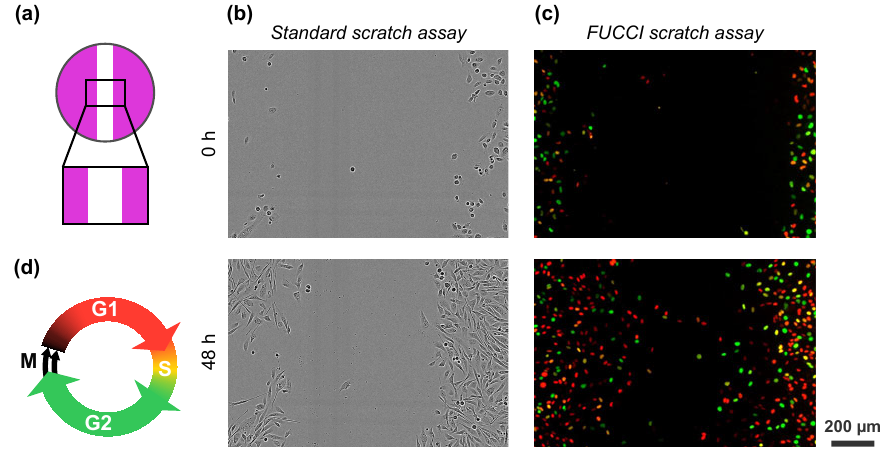}
		\caption[Fig 1]{\textbf{Partially observed spatial data in a scratch assay experiment.} (a) \textit{Scratch assays} involve growing a monolayer of cells (purple) in a well before making an artificial \textit{scratch} or \textit{wound} (white) and imaging a central region. (b) Snapshots collected at \SI{0}{\hour} and \SI{48}{\hour} from a typical scratch assay experiment conducted using PC3 prostate cancer cells \cite{Browning:2020}. Averaging in the direction of the scratch provides one-dimensional spatio-temporal information relating to the total cell density. (c) Scratch assay experiment with cell cycle information conducted using WM983C melanoma cells \cite{Carr.2021}. Encoding cells with FUCCI technology allows additional visualisation of the cell cycle \cite{Sakaue.2008}, shown schematically in (d). FUCCI scratch assay experiments provide spatio-temporal information relating to the density of cells in each stage of the cell cycle. Panels in (b) and (c) are reprinted from \cite{Browning:2020} and \cite{Carr.2021}, respectively, under a CC-BY and CC-attribution licence.}
		\label{fig1}	
	\end{figure}

\section{Methods and Results}

\subsection{Linear cell-cycle model}\label{two-state-cell-cycle}

	To demonstrate the differential algebra approach for linear PDE models, we first consider a two-state cell cycle model of cell migration subject to exponential growth \cite{Vittadello:2018,Simpson:2020}. We divide the cell cycle (\cref{fig1}d) into two subspecies that approximately correspond to the fluorescent markers in a FUCCI \cite{Sakaue.2008} scratch assay (\cref{fig1}c): cells in G1 phase fluoresce red, and cells in S/G2/M phase fluoresce green \cite{Simpson:2020}. During M-phase (mitosis), a green cell proliferates into two daughter cells that eventually fluoresce red. Finally, we model cell motility using linear diffusion. The model is given by
		\begin{subequations}\label{eq-cellcycle-lin1}%
		\begin{align}
			\pdv{r}{t} &= D_1\pdv[2]{r}{x} - k_1 r + 2k_2 g,\\
			\pdv{g}{t} &= D_2\pdv[2]{g}{x} + k_1 r - k_2 g,\label{eq-cellcycle-lin1-g}
		\end{align}
		\end{subequations}
	for $t > 0$ and $x \in [0,L]$, and where $r(x,t)$ and $g(x,t)$ denote the density of cells in G1 (red) and G2 (green), respectively; $D_1$ and $D_2$ are the diffusion coefficients associated with each subspecies; $k_1$ is the G1 (red) to G2 (green) transition rate; and $k_2$ the mitotic rate of green cells. For compactness, from this point on we denote derivatives using superscripts: $r^{(i,j)} = \partial^{i+j}r/\partial x^i \partial t^j$ and denote the undifferentiated variables interchangeably without $r^{(0,0)} \equiv r$. We prescribe homogeneous Neumann (no-flux) boundary conditions on a domain $x \in [0,L]$, such that no additional information is available from the equations that govern the behaviour at the boundaries. The initial condition is assumed to be general: $r(x,0) = r_0(x)$ and $g(x,0) = g_0(x)$. 

	While the model is motivated by FUCCI experiments that provide information on the cell-cycle status of individual cells (\cref{fig1}b), more typical experimental data is of the form 
		\begin{equation}\label{eq-lin-obs-n}
			n(x,t) = r(x,t) + g(x,t),
		\end{equation}
	as shown in \cref{fig1}c \cite{Johnston:2015}. A key question of interest is: can the model parameters, including the transition rates between cell-cycle states that are not individually observed, be inferred from an observation process of the form given by \cref{eq-lin-obs-n}; i.e., by measuring only the total cell density? 
	
	Application of the differential algebra approach to assess this question of structural identifiability requires first that we write the system \cref{eq-cellcycle-lin1} as a differential-algebraic equation in terms only of the variable that we observe, $n$, and its derivatives. To do so, we first write 
			\begin{equation}\label{eq-cellcycle-lin1-n}
				n^{(0,1)} = D_1 n^{(2,0)} + (D_2 - D_1)g^{(2,0)} + k_2 g^{(0,0)},
			\end{equation}
		to eliminate $r$. At this point, it is clear that the analysis is far simpler in the case that $D_2 = D_1 = D$: all that remains in this case is to solve \cref{eq-cellcycle-lin1-n} for $g$ in terms of $n^{(0,1)}$ and $n^{(2,0)}$, and substitute into \cref{eq-cellcycle-lin1-g}. For distinct diffusivities, progress is made by solving \cref{eq-cellcycle-lin1-n,eq-cellcycle-lin1-g} simultaneously for $g^{(0,1)}$ and $g^{(2,0)}$ in terms of $n$ and its derivatives to obtain
			\begin{subequations}
			\begin{align}
				g^{(0,1)} 	  &= k_1 n^{(0,0)} + \dfrac{\big[D_2(k_1 + 2k_2) - D_1(k_1 + k_2)\big]g^{(0,0)} - D_2 n^{(0,1)} + D_1D_2 n^{(2,0)}}{D_1 - D_2},\label{eq-cellcycle-lin1-gt}\\
				g^{(2,0)} &= \dfrac{k_2 g^{(0,0)} - n^{(0,1)} + D_1 n^{(2,0)}}{D_1 - D_2}.\label{eq-cellcycle-lin1-gxx}
			\end{align}
			\end{subequations}
	This step is advantageous as yields a first order equation, giving greater flexibility in subsequent equations that can be produced up to a specified order. Importantly, we note both the previous and all following steps may be carried out using only multiplication and elimination to avoid division by zero in determining the requisite first-order equation.
	
	Next, we differentiate \cref{eq-cellcycle-lin1-gt,eq-cellcycle-lin1-gxx} with respect to $x$ twice and $t$, respectively, to obtain two expressions for $g^{(2,1)}$ that can be equated to give, after some simplification,
		\begin{equation}\label{rd-lin-inputoutput}
			0 = -k_1k_2 n^{(0,0)} + (k_1 + k_2) n^{(0,1)} + n^{(0,2)} - (D_2 k_1 + D_1 k_2)n^{(2,0)} - (D_1 + D_2) n^{(2,1)}+ D_1 D_2 n^{(4,0)}.
		\end{equation}

	At this point, many authors divide through by the coefficient of the ``highest-order'' term, following some ordering, to ensure that the resultant polynomial system is monic, and therefore unique. A pathological example that illustrates why we must do this is to consider that we could otherwise multiply \cref{rd-lin-inputoutput} through by an arbitrarily chosen parameter. This would give the impression that the chosen parameter is identifiable through the coefficient of $n^{(0,2)}$. We take a more practical approach to ensure uniqueness and divide the resultant polynomial through by an arbitrarily chosen coefficient. For \cref{rd-lin-inputoutput}, we trivially divide through by the coefficient of $n^{(0,2)}$. Thus, the following set of polynomial coefficients (multiplied through by $-1$ where appropriate), are structurally identifiable
		\begin{equation}
			\Big\{k_1k_2,k_1+k_2,D_2k_1+D_1k_2,D_1+D_2,D_1D_2\Big\}.
		\end{equation}
	Clearly, since both the product $k_1k_2$ and sum $k_1 + k_2$ are structurally identifiable, so too are the individual parameters $k_1$ and $k_2$. A similar observation can be made for the diffusivities, and so all model parameters are structurally identifiable. As we expect, removing spatial information (for example, by initialising the system with a spatially homogeneous initial condition such that $n^{(i,0)} = 0$ for all $i = 1,2,\ldots$), the polynomial system reduces to terms that do not include $D_1$ nor $D_2$, however the rate constants $k_1$ and $k_2$ remain identifiable.

\subsection{Generic two-state reaction-advection-diffusion system}\label{generic_two_state}
	
	Having established the general procedure for applying the differential algebra framework to reaction-diffusion models, we now consider the general two-state model
	\begin{subequations}\label{generic_rad}
	\begin{align}
		u^{(0,1)} &= D_u u^{(2,0)} + \alpha_u u^{(1,0)} + p_1 u + p_2 v + p_3,\\
		v^{(0,1)} &= D_v v^{(2,0)} + \alpha_v v^{(1,0)} + p_4 u + p_5 v + p_6\label{generic_rad_vt},
	\end{align}
	\end{subequations}
	for $t > 0$ and 	subject to observations of the form $n(x,t) = u(x,t) + v(x,t)$. The initial conditions are fully prescribed by $n(x,0) = n_0(x)$ and $v(x,0) = v_0(x)$. We do not consider a domain or boundary conditions, only assuming that no information is available through observation of the system at the boundaries.

	Analysis of such a system captures both extensions to the linear cell-cycle model that incorporates advection, and canonical models analysed in the structural identifiability literature such as the so-called ``two-pool model'' \cite{Cobelli:1980}. The following results are trivially applicable to other linear combinations of the states (for example, observations of the form $\alpha u(x,t) + \beta v(x,t)$ where $\alpha$ and $\beta$ are unknown parameters) through a rescaling.
	
	The procedure for determining the required polynomial equation is now complicated by the presence of first order spatial derivatives $v^{(1,0)}$ in the system. To proceed, we first eliminate $u$ to obtain a two-state system in $n$ and $v$, with
		\begin{equation}\label{generic_rad_n}
		\begin{aligned}
			n^{(0,1)} &= D_u n^{(2,0)} + (D_v - D_u) v^{(2,0)} + \alpha_u n^{(1,0)} + (\alpha_v - \alpha_u) v^{(1,0)}\\&\qquad+ (p_1 + p_4)n^{(0,0)} + (p_2 + p_5 - p_1 - p_4) v^{(0,0)} + p_3 + p_6.
		\end{aligned}
		\end{equation}
	We proceed by solving \cref{generic_rad_vt,generic_rad_n} for $\left(v^{(2,0)},v^{(0,1)}\right)$ so as to obtain a system of one first order and one second order equation. We can then expand the resultant system by taking appropriate derivatives up to $i$th order to end up with a linear system with more equations than unknown variables. In this case, this is achieved by considering derivatives up to order $i = 4$, which yields 16 equations in all 15 possible fourth order derivatives of the unobserved variable $v$. The resultant system will \textit{always} be linear for all linear and what we term ``semi-linear'' systems, the latter defined as systems that are linear in the variables we wish to eliminate (a logistic term of the form $v(1 - n)$ would be semi-linear, since $n$ is observed). 
	
	Thus, we arrive at the over-determined system
		\begin{equation}
			\mathbf{A} \mathbf{v}^{(4)} = \mathbf{b}, 	
		\end{equation}
	where $\mathbf{v}^{(i)} \in \mathbb{R}^{(i+1)(i+2)/2}$ denotes a vector of all possible $i$th order derivatives. We then proceed by performing Gaussian elimination to reduce the augmented matrix $(\mathbf{A} | \mathbf{b})$ into row echelon form. This provides a set of expressions, based upon linear combinations of the elements of $\mathbf{b}$, that necessarily vanish; thus providing the required set of polynomial equations.
	
	{It is not necessarily the case that we must work with the full set of fourth order equations,  just that we are required to do so to ensure that the system is appropriately determined. This is important as it is possible for the system to be characterised by more than one polynomial equation: the number of polynomial equations we require for linear and semi-linear problems will be given by the difference in rank of $\mathbf{A}$ and the rank of the augmented matrix $(\mathbf{A} | \mathbf{b})$ from the over-determined system. For the two-state RAD system, we thus require a single polynomial equation. It happens in our case that we are able to obtain this single polynomial equation from a subset of third order equations. Algebraic manipulations are performed in \texttt{Mathematica} \cite{Mathematica} and available as supplementary code with algorithmic details provided in \Cref{appalg}.}
	
	For the system defined by \cref{generic_rad}, we arrive at the set of polynomial coefficients 
	\begin{equation}\label{generic_rad_ident_combo}
	\begin{aligned}
		\Big\{p_1p_5 - &p_2p_4,p_1 + p_5,p_5 \alpha_u + p_1 \alpha_v,\alpha_u + \alpha_v,D_v p_1 + D_u p_5 + \alpha_u \alpha_v, \\& -D_u - D_v,D_v\alpha_u + D_u \alpha_v,D_uD_v,(p_5-p_4)p_3+(p_1 - p_2)p_6\Big\}.
	\end{aligned}
	\end{equation}

	{eq10}\Cref{generic_rad_ident_combo} provides an exhaustive set of identifiable parameter combinations. We then reduce this set further to establish a fully-reduced set of identifiable parameters and parameter combimnations. For example, denoting two combinations by $c_6 = D_u + D_v$ and $c_8 = D_u D_v$ we can solve for the parameters $D_u$ and $D_v$ in terms of the identifiable parameter combinations $c_6$ and $c_8$, thus establishing that $D_u$ and $D_v$ are identifiable. For \cref{generic_rad}, the reduced set of identifiable parameter combinations is given by
		\begin{equation}\label{generic_rad_idcombo}
		\begin{aligned}
			\Big\{D_u,D_v,\alpha_u,\alpha_v,p_1,p_5,p_2p_4,(p_5-p_4)p_3+(p_1 - p_2)p_6\Big\}.
		\end{aligned}
		\end{equation}

	By setting $\alpha_u = \alpha_v = 0$ in \cref{generic_rad_ident_combo}, we see that these results are identical to the no-advection model. However, results differ slightly for the non-spatial model. Setting $D_u = D_v = \alpha_u = \alpha_v = 0$ to remove spatial derivatives, we see that the reduced set of identifiable parameter combinations is now
		\begin{equation}
		\begin{aligned}
			\Big\{p_1p_5 - p_2p_4,p_1 + p_5,(p_5-p_4)p_3+(p_1 - p_2)p_6\Big\},
		\end{aligned}
		\end{equation}
	so we can no longer identify $\textit{any}$ rate parameters individually. This result demonstrates that we are potentially able to learn more about the process by introducing spatial heterogeneity into the system.

	\subsection{Role of the initial condition in identifiability}\label{role_of_ic}
	
	Clearly, if we observe only $n(x,t) = u(x,t) + v(x,t)$ we cannot always fully ascertain the initial condition $v_0(x)$. A key point of distinction between such partially observed PDE models and their ODE counterparts is that the initial condition is an unknown $\textit{function}$ whereas for an ODE model the initial condition is merely an additional unknown scalar parameter. In the case of structural identifiability, we assume that we perfectly observe not only $n(x,0) = n_0(x)$ but also its derivatives. Thus, we can consider that the right-hand-side of \cref{generic_rad_n} is ``observed'' such that the functional 
		\begin{equation}\label{generic_rad_poissonic}
		\begin{aligned}
			f(x) := n_0^{(0,1)}(x) &= D_u n_0^{(2,0)}(x) + (D_v - D_u) v_0^{(2,0)}(x) + \alpha_u n_0^{(1,0)}(x) + (\alpha_v - \alpha_u) v_0^{(1,0)}(x)\\&\qquad+ (p_1 + p_4)n(x) + (p_2 + p_5 - p_1 - p_4) v_0(x) + p_3 + p_6,
		\end{aligned}
		\end{equation}
	is also ``identifiable'' in a similar sense to the other identifiable parameter combinations in \cref{generic_rad_idcombo}. This result shows a clear interdependence between model parameters and the initial condition: specifiying a functional form for the initial condition has a similar role to fixing a non-identifiable parameter that appears in combination with other model parameters. Thus, it is possible for a practitioner to unknowingly render structurally non-identifiable parameters identifiable through an assumption related to the initial condition, highlighting the importance of considering structural identifiability when working with partially observed PDE models.

	\begin{figure}
		\centering
		\includegraphics[width=\textwidth]{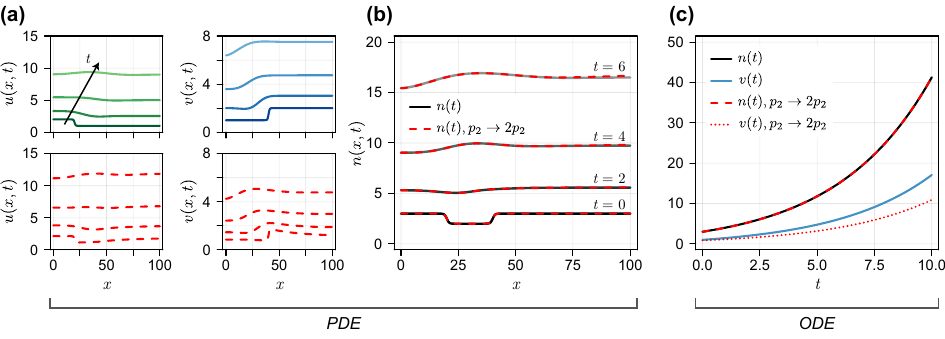}
		\caption[Fig 2]{\textbf{Non-identifiability of the linear RAD model.} (a) Individual (unobserved) states, $u(x,t)$ and $v(x,t)$, at $t = 0$, 2, 4 and 6. The arrow indicates the direction of increasing $t$. Shown are curves for the original set of parameter values (green and blue solid curves, top row) and curves for the modified set of parameter values where $p_2 \mapsto 2p_2$ (red dashed, bottom row). Note the difference in initial condition between parameter sets. (b) Observed variable, $n(x,t) = u(x,t) + v(x,t)$ at the original set of parameter values (grey) and modified set of parameter values (red dashed). Note that the parameter sets are indistinguishable. (c) Equivalent dynamics for the ODE model, showing $n(t)$ and $v(t)$ at the original parameter set (solid) and modified parameter set (dashed). Original parameters are $D_u = 20$, $D_v = 10$, $\alpha_1 = -5$, $\alpha_2 = 5$, $p_i = 0.1$ $(i \neq 2)$ and $p_2 = 0.2$.}
		\label{fig2}	
	\end{figure}
	
	We illustrate these results in \cref{fig2} by noting that \cref{generic_rad_idcombo} gives combinations of parameters that can produce identical model outputs. For example, we can set $p_2 \mapsto 2p_2$ and see no change to $n(x,t)$ provided other parameters and the initial condition are adjusted accordingly. Solving the model with an initial parameter set and initial condition for $n(x,0)$ provides $f(x)$ to determine the new initial condition for $v(x,t)$ in terms of \cref{generic_rad_poissonic}, a steady-state reaction-advection equation (\cref{fig2}a). Solving the full system again (\cref{fig2}b) demonstrates structural non-identifiability: model outputs are indistinguishable. We show the equivalent set of results for the corresponding ODE model in \cref{fig2}c.

\subsection{Generic linear systems}\label{generic_linear_systems}
		
	A question naturally arising from our analysis of the generic two-state system is whether the procedure is applicable to linear systems with an arbitrary number of states and outputs; a second is whether the \textit{reduction} in the number of identifiable parameter combinations when reducing to a corresponding ODE model is generally true of RAD models.
	
	To determine whether it is, in theory, possible to find a polynomial system for all $m$-state linear RAD models subject to observations of $\ell$ variables, we need only consider whether it is possible to \textit{close} such a system through repeated differentiation. That is, is there an order of derivative at which we are guaranteed that we will have more determining equations than variables? We arrive at the following result.
	
	\begin{theorem}\label{theorem1}
		All linear RAD models of $m$ states can be reduced to a set of polynomial relations involving derivatives of order no more than $4(m-1)$.
	\end{theorem}
	\begin{proof}
		We can rewrite the system of $m$ equations to include at least one first-order equation. Expanding the system to include derivatives up to order $n$ yields at least $q(n) = (mn^2 - mn + 2n) / 2$ equations, with no more than $v(n) = (n+1)(n+2)(m-1)/2$ unknown variables (i.e., $n$th order and lower derivatives of the unobserved quantities). The number of excess variables is no more than $d(n) = q(n) - v(n)$, which we require to be positive. Defining $n^* = \min n: d(n) \ge 0, n \in \mathbb{Z}$, it is straightforward to see that $n^* = 4(m-1)$, as required.
	\end{proof}
	
	\Cref{theorem1} highlights the potential computational cost of performing structural identifiability for PDE models with large numbers of states: application of our method for a linear system of only three states would potentially require us to perform Gaussian elimination symbolically on a 92 equation system. While analysis of the two-state system required a walltime of approximately \SI{1}{\second}, attempting to analyse a corresponding generic three-state model subject to a single observation did not yield results within a \SI{12}{\hour} runtime (\texttt{Mathematica 13.1} \cite{Mathematica}, Apple M1 Pro with 16GB RAM).	Importantly, however, \cref{theorem1} demonstrates that it is always theoretically possible to perform structural identifiability analysis on linear RAD models. Furthermore, the Gaussian-elimination-based reduction procedure requires only that the system be linear in the unobserved quantities, thus \Cref{theorem1} will hold for all semi-linear systems, including, for example, analogues of \cref{eq-cellcycle-lin1} with logistic growth terms and observation processes that include the total population.
	
	Secondly, we observe for the linear two-state system that the PDE model contains a greater number of identifiable parameter combinations than the corresponding ODE (i.e., spatially homogeneous) model. We present the following result.
	
	\begin{corollary}\label{theorem2} Identifiable parameter combinations in a linear spatially homogeneous system are a subset of the identifiable parameter combinations in all corresponding spatially heterogeneous models subject to linear advection and/or diffusion with general unknown initial conditions.
	\end{corollary}
	\begin{proof}
	By \Cref{theorem1} we are guaranteed that the spatially heterogeneous system may be written as set of polynomial relations, and, therefore, that we can express the system as a set of polynomial relations involving only observed variables and their derivatives. Thus, we can obtain a set of identifiable parameter combinations, $\mathcal{Q}$, from the set of polynomial coefficients. Next, we can recover the set of identifiable parameter combinations in the spatially homogeneous model, $\mathcal{Q}^*$, by considering that all terms involving spatial derivatives vanish. Thus, $\mathcal{Q}^* \subseteq \mathcal{Q}$ as required.
	\end{proof}
	
	As stated previously, we will also be able to write semi-linear systems as over-determined linear system of differential-algebraic equations, and consequentially \Cref{theorem2} will also hold for semi-linear systems. Furthermore, we propose that one cannot increase the set of identifiable parameter combinations by imposing additional constraints on the problem (for example, by prescribing a particular form for the initial condition). Thus, we also expect \Cref{theorem2} to hold for systems with non-linear source terms by considering identifiability of systems with no-flux boundary conditions and homogeneous initial conditions, which reduce to the corresponding set of ODEs.

\subsection{Semi-linear systems}\label{semi-linear-systems}

	We now briefly present and discuss analysis of two specific semi-linear systems: a two-state model of bacteria migration due to chemotaxis, and a three-state cell-cycle model subject to logistic growth. 

	\subsubsection{Two-state model of bacteria chemotaxis}

	First, we consider a model of bacteria chemotaxis. Bacteria, $\rho(x,t)$, secrete a chemotactic factor, $c(x,t)$, at rate $k > 0$. Both bacteria and the chemotactic factor diffuse, however bacteria additionally undergo directed motion due to the spatial gradients in the chemotactic factor, with strength and direction given by $\chi$ ($\chi > 0$ and $\chi < 0$ correspond to negative and regular chemotaxis, respectively). The chemotactic factor degrades at constant rate $\alpha > 0$. The dynamics are described by the non-linear system \cite{Murray:2002}
	\begin{subequations}\label{pde_chemotaxis}
	\begin{align*}
		\rho^{(0,1)} &= D_\rho \rho^{(2,0)} + \overbrace{\chi(\rho c^{(0,1)})^{(0,1)}\vphantom{\Big(}}^{\text{Chemotaxis}},\\
		c^{(0,1)} &= D_c c^{(2,0)} + k\rho - \alpha c,
	\end{align*}
	\end{subequations}
	for $t > 0$ and $x \in [0,L]$. We consider a realistic experimental setup where only information relating to the density of bacteria is available (i.e., the concentration of the chemotactic factor is not observed). The experiment is initiated without the chemotactic factor, $c(x,0) = 0$, and the field of view is such that no-flux boundary conditions are appropriate. We consider the system to be semi-linear as it is linear the unobserved quantity, $c$.
	
	Results in \Cref{theorem1} indicate that, after eliminating $c^{(2,0)}$ from one equation, we can expand the system to include derivatives up to fourth order to ensure that the system is fully determined. Performing row reduction on the expanded system yields a polynomial expression containing 212 coefficients, which determine the identifiable parameter combinations (full calculations available as supplementary code) as
		\begin{equation}
			\left\{D_\rho,D_c,k,\alpha \chi \right\}.
		\end{equation}
	Thus, $D_\rho$, $D_c$, and $k$ are structurally identifiable as is the product $\alpha\chi$, while the individual constituents $\alpha$ and $\chi$ are not. This analysis agrees with results that one would arrive through the scaling $\tilde{c} = \alpha c$ to eliminate the disparate appearance of $\alpha$ and $\chi$ in \cref{pde_chemotaxis}.
	
	A key difference between analysis of the chemotaxis model and the general linear model considered previously is the role of the initial condition. While it may be true that $\alpha$ and $\chi$ would become individually identifiable if information relating to a non-homogeneous non-zero initial condition were known, the initial condition $c(x,0)$ conveys no information about the parameters. Indeed, we can demonstrate that $\alpha$ and $\chi$ can compensate for each other without varying the initial condition (\cref{fig3}).

	\begin{figure}
		\centering
		\includegraphics{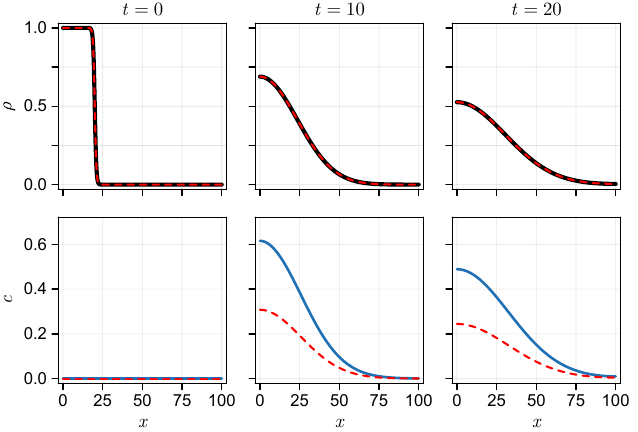}
		\caption[Fig 3]{\textbf{Non-identifiability of the chemotaxis model.} Solutions show the observed variable $\rho(x,t)$ (top row), and unobserved variable $c(x,t)$, at the original parameter set (solid curves) and modified parameter set where $\chi \mapsto 2\chi$ (red dashed). Original parameters are $D_\rho = 20$, $D_c = 100$, $\chi = k = \alpha = 1$. Note that the choice $\chi > 0$ prescribes \textit{negative chemotaxis}: bacteria move down the chemotactic gradient.}
		\label{fig3}	
	\end{figure}

\subsubsection{Three-state cell-cycle model of cell migration subject to logistic growth}

	Finally, we consider a three-state cell cycle model of cell migration that extends on that presented in \cref{two-state-cell-cycle} by capturing the intermediate stage where cells appear to fluoresce yellow in FUCCI assays, and where contact inhibition prevents cells in G2 from undergoing mitosis \cite{Vittadello:2018,Donker.2022}. The model is given by
	\begin{subequations}\label{fucci3}
	\begin{align}
		r^{(0,1)} &= D_1 r^{(2,0)} - \lambda_1 r + 2\lambda_3g(1 - n/K),\label{fucci3_r}\\
		y^{(0,1)} &= D_2 y^{(2,0)} - \lambda_2 y + \lambda_1 r,\label{fucci3_y}\\
		g^{(0,1)} &= D_3 g^{(2,0)} - \lambda_3 g(1 - n/K) + \lambda_2 y.\label{fucci3_g}
	\end{align}
	\end{subequations}
	Here, the variable $y$ corresponds to cells that fluoresce both red and green and hence appear to fluoresce yellow. We consider that FUCCI scratch assay experiments are conducted (\cref{fig1}b), but that we do not attempt to distinguish the period of fluorescent overlap so that only measurements relating to $r$ and $g$ are available.
	
	While this problem seems more complex---both due to the non-linearity introduced by the logistic growth term and the additional state variable---it is relatively straightforward to eliminate $y$ and its derivatives without expanding the system through differentiation. To do this, we solve \cref{fucci3_y} explicitly for $y$ to obtain
		\begin{equation*}
			y = \dfrac{-\lambda_3 (g^{(0,0)})^2 + \lambda_3 g^{(0,0)} (K - r^{(0,0)}) + K \big(g^{(0,1)} - D_3 g^{(2,0)}\big)}{K \lambda_2 + \lambda_3 g^{(0,0)}},
		\end{equation*}
	and then substitute the resultant expression directly into \cref{fucci3_r,fucci3_y}. This approach is advantageous in that it can be applied to \textit{any} non-linear partial differential equation model in which we can solve for the unobserved quantities in terms of the observed quantities and respective derivatives; we could, for example, apply such a technique to a model where growth inhibition is modelled through alternative growth models such as Gompertz or Richards (for models with non-polynomial terms the orthogonality of terms will also need to be established to construct the set of coefficients). Expanding the resultant system reveals that all parameters are structurally identifiable, in agreement with the result in \Cref{theorem2}, where we find that the corresponding ODE model is also structurally identifiable according to the online tool SIAN \cite{Hong:2019uq,Hong.2020}.

\section{Discussion and Conclusion}\label{conclusion}

Many processes are inherently spatial, and not well described by ODE models. Yet, tools for assessing the identifiability of the PDE models that are increasingly used to interpret spatial data are almost entirely undeveloped. In this brief article, we illustrate how the well-established differential algebra approach to structural identifiability of ODE models can be transferred to analyse a large class of spatial models. 

While we demonstrate that the differential algebra approach can be applied to any system of linear RAD equations, we have restricted our demonstration to two- and three-state systems in one spatial dimension. Results in \Cref{theorem1} show that the derivative order required grows rapidly as the number of state variables increases; we expect this to be exacerbated for PDE models with more than one spatial dimension. This presents not only computational difficulties, but practical issues with interpretation: the resultant set of polynomial coefficients will be verbous and potentially complicated, making it difficult to establish a reduced set of identifiable parameter combinations. We see this even for the relatively simple semi-linear chemotaxis model (supplementary code). Such computational issues are a well-established shortcoming of the differential algebra approach, even for ODE models \cite{Raue:2014}. While this bottleneck presents a clear shortcoming, many PDE models used in practice only possess a small number of state variables: perhaps more so than for ODE models, where large systems are common \cite{Karlsson.2012}. Common systems of PDE models that we can analyse include models of chemotaxis \cite{Murray:2002}, Turing patterning and morphogenesis \cite{Turing:1952}, age-structured epidemic models \cite{Renardy.2022}, systems of heat and wave equations, and many more. For some systems, one approach to deal with larger numbers of state variables is through a transformation that decouples states \cite{Sun.1999de}; identifiability can then be established for species sequentially. 

Another approach to alleviate the computational cost is to develop more efficient algorithmic implementations of state-variable elimination that do not involve expanding to a fully determined system. In the ODE literature, Ritt's algorithm is commonly employed alongside a Gr\"{o}bner basis factorisation to eliminate state variables automatically \cite{Bellu:2007,Meshkat.2012}. An advantage of eventual implementation of these algorithms for analysis of PDE models is their application to non-linear problems: computing the Gr\"{o}bner basis is equivalent to Gaussian elimination for linear systems of equations, but can be applied generally to any polynomial system. While it remains unclear whether this can, given infinite computational time, always be applied to PDE models (as we have shown is the case for linear systems), such software will enable analysis of a much larger class of PDE models that can be written as polynomial in both state variables and their corresponding derivatives \cite{Rahkooy:2024}. 

A unique feature in the application of PDE models to interpret data is the mode of measurements that may be collected. It is not always the case that experimental observations correspond to measurements of a spatiotemporal function: rather, observations could comprise scaler measurements at a single point in space (for example, measurements of temperature by a probe in a heat conduction experiment \cite{McInerney.2019rs}) or scalar measurements of a spatial average (i.e. of total cell count). Summary-statistic type measurements can be even more complicated for models with two or more spatial dimensions \cite{Browning.2021}. The scratch assay experiment is an example of this, as data collected from the full, two-dimensional, process often comprises spatial averages in the direction parallel to the scratch \cite{Johnston:2015}. This example is trivial as, provided appropriate constraints are placed on the initial condition, the extraneous spatial dimension can be integrated out of the full two-dimensional model to yield the one-dimensional model that we analyse. It is entirely unclear whether it is possible, in the general case, to reduce a system of partial differential equations to a polynomial system involving an observation function that does not include space. 

This question of lower-dimensional observation functions is particularly relevant to structural identifiability analysis of stochastic differential equation models \cite{Browning.2020}. It is unclear whether such equations can be reduced to eliminate unobserved states, since the underlying stochastic process is typically not differentiable. One approach is to analyse the corresponding system of Fokker-Plank equations, a PDE in as many spatial dimensions as SDE state variables. Partially observing the system, therefore, corresponds to observing marginals of the Fokker-Plank equation solution; in effect, integrating over unobserved spatial variables.

PDE models are widely used to characterise spatial processes and interpret spatial biological data. It is, therefore, paramount that effective and efficient tools to assess the structural identifiability of these models be developed. We build on existing work \cite{Renardy.2022,Ciocanel.2023} to show that the differential algebra approach to structural identifiability analysis can be applied to all linear RAD PDE models, and some classes of non-linear PDE models. We demonstrate the interdependence between structural identifiability and the initial conditions in partially observed models, highlighting the importance of assessing the structural identifiability before attempting to infer parameters in PDE models from experimental data.

\section*{Data availability}
	Mathematica and Julia code used to perform the symbolic and numerical computations, respectively, available on GitHub at \url{https://github.com/ap-browning/pde_structural_identifiability}.

\section*{Acknowledgements}
	APB thanks the Mathematical Institute for a Hooke Fellowship. MT and APB thank the EPSRC and Mathematical Institute for funding a vacation internship that funded research conducted by MT. CF acknowledges support of a fellowship from ``la Caixa'' Foundation (ID 100010434) with code LCF/BQ/EU21/11890128. REB is supported by a grant from the Simons Foundation (MP-SIP-00001828).

\appendix
\clearpage
\refereecomment{R2.3}{algorithm}
\section{Algorithm of implementation}\label{appalg}

In \cref{alg1} we provide an algorithm that outlines our implementation of the differential algebra approach to assess the structural identifiability of $m$-state RAD PDE subject to $\ell$ linearly independent observations. An implementation of the analysis in \texttt{Mathematica} \cite{Mathematica} is available as supplementary material on GitHub (\url{https://github.com/ap-browning/pde_structural_identifiability}).

\vfill
\begin{algorithm}[H]
\caption{Differential algebra framework for linear and semi-linear RAD PDE models}
\label{alg1}
\begin{algorithmic}[1]
	\item Given a system of $m$ equations, denoted $\mathcal{Y}$, and a set of $\ell < m$ linearly-independent observed quantities. 
	\item Rewrite the system $\mathcal{Y}$ in terms of only the $\ell$ observed quantities and a remaining set of $k = m - \ell$ unobserved quantities that are to be eliminated. The rewritten system of $m$ equations is denoted by $\mathcal{Y}_1$. 
	\item Reduce $\mathcal{Y}_1$ such that the first equation is first-order in the unobserved quantities. The reduced system is denoted by $\mathcal{Y}_2$. This can be done by, for example, solving the last $m-1$ equations simultaneously for the second-order spatial derivatives of the unobserved quantities and substituting the resultant expressions into the first equation.
	\item Apply \Cref{theorem1} to determine that the required order of the expanded system is at most $n^* = 4(m - 1)$. 
	\item Expand $\mathcal{Y}_2$ up to order $n^*$ by taking all possible order $n^*-1$ partial derivatives of the first equation, and all possible order $n^*-2$ partial derivatives of the remaining equations. The expanded system, denoted by $\mathcal{Y}_3$ will comprise $\tilde{m}$ equations, linear in a total of $\tilde{n} < \tilde{m}$ partial derivatives of the unobserved variables. 
	\item Write the expanded system as the linear system $\mathbf{A}\mathbf{v}^{(n^*)} = \mathbf{b}$, where $\mathbf{A} \in \mathbb{R}^{(\tilde{m},\tilde{n})}$ and $\mathbf{b} \in \mathbb{R}^{(\tilde{m})}$.
	\item Perform Gaussian elimination to reduce the symbolic augmented matrix $(\mathbf{A} | \mathbf{b})$ into row-echelon form, denoted $(\mathbf{A}_\text{re}|\mathbf{b}_\text{re})$.
	\item The set of polynomial equations, denoted by $\mathcal{R}$, is given by the non-zero elements of $\mathbf{b}_\text{re}$ corresponding to rows of $\mathbf{A}_\text{re}$ that are identically zero. Elements of $\mathcal{R}$ are polynomial in the set of partial derivatives of the observed $\ell$ quantities and are also considered to be observed.
	\item Normalise each polynomial to ensure uniqueness by dividing through by a chosen non-zero coefficient. The set of monic polynomials is denoted by $\mathcal{R}_1$.
	\item The set of identifiable parameter combinations, $\mathcal{Q}$, is given by union of the the sets of polynomial coefficients of each element of $\mathcal{R}_1$. Coefficients that do not depend on the unknown quantities should be removed before further reduction.
\end{algorithmic}
\end{algorithm}
\vfill

 
\end{document}